\documentclass[11pt]{article}
\usepackage{graphicx}
\usepackage{amsfonts,amssymb,eucal,amsmath}

\newtheorem{theorem}{Theorem}

\newtheorem{corollary}[theorem]{Corollary}

\newtheorem{example}[theorem]{Example}

\newtheorem{proposition}[theorem]{Proposition}
\newtheorem{remark}[theorem]{Remark}

\newenvironment{proof}[1][Proof]{\noindent\textbf{#1.} }{\ \rule{0.5em}{0.5em}}
\numberwithin{theorem}{section}
\numberwithin{equation}{section}
\input{tcilatex}
\begin{document}

\title{Jet Theoretical Yang-Mills Energy in the Geometric Dynamics of
2D-Monolayer}
\author{M. Neagu$^{1}$, N.G. Krylova$^{2}$, H.V. Grushevskaya$^{3}$}
\date{}
\maketitle

\begin{abstract}
Langmuir-Blodgett films (LB-films) consist from few LB-monolayers which are
high structured nanomaterials that are very promising materials for
applications. We use a geometrical approach to describe structurization into
LB-monolayers. Consequently, we develop on the $1$-jet space $%
J^{1}([0,\infty ),\mathbb{R}^{2})$ the single-time Lagrange geometry (in the
sense of distinguished (d-) connection, d-torsions and an abstract
anisotropic electromagnetic-like d-field) for the Lagrangian governing the
2D-motion of a particle of monolayer. One assumed that an expansion near
singular points for the constructed geometrical Lagrangian theory describe
phase transitions to LB-monolayer. Trajectories of particles in a field of
the electrocapillarity forces of monolayer have been calculated in a
resonant approximation utilizing some Jacobi equations. A jet geometrical
Yang-Mills energy is introduced and some computer graphic simulations are
exposed.
\end{abstract}

%\\[0pt]

\begin{center}
%\textit{\textbf{M. Neagu}$^{1}$, \textbf{N.G. Krylova}$^{2}$, \textbf{H.V.
%Grushevskaya}$^{3}$} \\[0pt]
$^{1}$ Department of Mathematics and Informatics, University Transilvania of
Bra\c{s}ov, 50 Iuliu Maniu Blvd., 500091 Bra\c{s}ov, ROMANIA, E-mail:
mircea.neagu@unitbv.ro\\[0pt]
$^{2}$ Physics Faculty, Belarusian State University, 4 Nezavisimosti Ave.,
220030 Minsk, BELARUS, E-mail: nina-kr@tut.by \\[0pt]
$^{3}$ Physics Faculty, Belarusian State University, 4 Nezavisimosti Ave.,
220030 Minsk, BELARUS, E-mail: grushevskaja@bsu.by
\end{center}



\section{Introduction}

LB-films are formed at sequential transfer of LB-monolayers of amphiphilic
molecules on a solid support. The LB-monolayer having high-ordered structure
self-organizes on a surface of liquid subphase at its compressing by a
barrier. %Пленки Ленгмюра-Блоджетт формируются при
%последовательном переносе на твердую подложку  ЛБ-монослоев
%амфифильных молекул, которые самоорганизуются на поверхности
%жидкой субфазы при поджатии их барьером, образуя
%высокоупорядоченные структуры.
LB-technique allows to form films containing nanosize objects (quantum dots,
nanotubes, %Техника ЛБ позволяет формировать
%пленки, содержащие наноразмерные объекты ( квантовые точки,
%нанотрубки,
nanowires, nanorods), which become ordered in the process of monolayer
formation. %которые в процессе формирования
%монослоя упорядочиваются.
The resulting LB-films possess unique properties which can not be obtained
in %
%Получаемые ЛБ-пленки обладают
%уникальными свойствами, которые не могут быть получены в
"bulk", as they reveal nanoscale effects. % так
%как проявляют наноразмерные эффекты.
Such objects find wide applications in optics and nanoelectronics \cite%
{Acharya[6]}, %Эти объекты находят широкое
%применение в оптике и наноэлектронике. [S.Acharya,J.P.Hill,
%K.Ariga Soft Langmuir-Blodgett technique for hard nanomaterials /
%Adv.Mater. 2009, 21, 2959-2981]
\cite{Motshmann[7]} %[H.Motschmann, H.Mo:hwald
%Langmuir-Blodgett films / Handbook of applied surface and colloid
%chemistry / chapter 28 / Ed. by K.Holmberg - 2001. - 629-648.]

In the process of LB-monolayer formation on a subphase surface one observes
a phase transition of the first order from 2-dimensional gas up to condense
state of monolayer. 
%В процессе образования ЛБ-монослоя на поверхности субфазы
%наблюдаются последовательные фазовые переходы 1-ого рода от
%2-мерного газа до конденсированного состояния монослоя.
Up today, a theoretical description of the phase transition has been
performed %by
%В
%настоящее время теоретическое описание фазовых переходов было
%проведено by
%V.M.~Kaganer {\it et al.}
with the use of the Landau theory %с использованием теории Ландау
\cite{Kaganer}. Nevertheless, such phenomenological description does not
allow to obtain some experimentally observable characteristics of
structurization process. %Тем не
%менее это феноменологическое описание не позволяет получить ряд
%экспериментально наблюдаемых характеристик параметров
%структурообразования.
%For example, isotherms of compressing of monolayer do not
%demonstrate horizontal plateau as in the bulk case, at
%phenomenological description this is linked with different
%orientation of molecules' tails in the formed LB-monolayer.
%Например, изотермы сжатия монослоя не имеют
%горизонтального плато, как для случая сжатия вещества в объеме,
%что в феноменологическом описании связывают с различной
%ориентацией "хвостов" молекул в образованном ЛБ-монослое.
It was also shown that the structurization is observed in quasi-equilibrium
state of monolayer at definite values of compressing speed whereas at
intermediate values the phase transition is not observed \cite{Hrushevsky}. 
%Было
%показано, что упорядочение молекул наблюдается в квази-равновесном
%состоянии монослоя при определенных значений скорости сжатия, и
%между этими значениями фазовый переход не наблюдался.
At this, typical dependence of surface tension $P$ upon the surface area $S$
was observed for such values of tightening speed when the compression
process occurs far beyond the equilibrium. %При этом
%типичная зависимость поверхностного натяжения ?Пи от площади
%поверхности S для фазовых переходов первого рода наблюдалась для
%таких скоростей сжатия, при которых процесс сжатия происходит
%вдали от равновесного состояния.
For the correct description of such peculiarities of the process it is
necessary to use kinetic and dynamical models 
%Для корректного описания таких
%особенностей процесса необходимо привлекать кинетические и
%динамические модели.

In the paper \cite{Grushevskaya2011} an examination of LB-monolayer
structurization process using differential geometry methods has been
proposed and the corresponding 3D Finsler metric of space in which a
particle of monolayer is moving has been constructed. An obtained Lagrangian
will be used further. 
%Ранее в работе [] было предложено рассмотрение процессов
%структурообразования в ЛБ-монослоях с использованием методов
%дифференциальной геометрии и построена метрика 3-мерного
%пространства монослоя (r,fi,t).

Let us start with the usual physical time defined by the Euclidian manifold $%
(T=[0,\infty ),h_{11}(t)=1),$ whose Christoffel symbol is $\kappa
_{11}^{1}(t)=0$. Let us also consider the plane manifold $\mathbb{R}^{2}$,
having the polar coordinates $\left( r,\varphi \right) $, where $r>0$ and $%
\varphi \in \lbrack 0,2\pi )$. Let us construct the $1$-jet vector bundle $%
J^{1}(T,\mathbb{R}^{2})\rightarrow \mathbb{R}\times \mathbb{R}^{2}$, locally
endowed with the coordinates $(t,x^{1},x^{2},y_{1}^{1},y_{1}^{2}):=\left(
t,r,\varphi ,\dot{r},\dot{\varphi}\right) $. We remind that on the $1$-jet
space $J^{1}(T,\mathbb{R}^{2})$ a local transformation of coordinates is
given by the rules (the Einstein convention of summation is used throughout
this work):\footnote{%
Throughout this paper the Latin letters $i,$ $j,$ $k,l,q,s,...$ take values
in the set $\{1,2\}$.}%
\begin{equation}
\widetilde{t}=\widetilde{t}(t),\quad \widetilde{x}^{q}=\widetilde{x}%
^{q}(x^{s}),\quad \widetilde{y}_{1}^{q}=\dfrac{\partial \widetilde{x}^{q}}{%
\partial x^{s}}\dfrac{dt}{d\widetilde{t}}\cdot y_{1}^{s},  \label{tr-rules-2}
\end{equation}%
where $d\widetilde{t}/dt\neq 0$ and rank $(\partial \widetilde{x}%
^{p}/\partial x^{q})=2$. According to Olver's opinion \cite{Olver}, we
believe that the $1$-jet spaces are natural houses for the study of many
physical theories. For such a reason, using the special function%
\begin{equation*}
f\left( z\right) \overset{def}{=}-\int_{-z}^{\infty }\frac{e^{-t}}{t}dt,
\end{equation*}%
we study the 2D-motion of a particle of monolayer governed by the jet
Lagrangian function $L:J^{1}(T,\mathbb{R}^{2})\rightarrow \mathbb{R}$
defined by%
\begin{equation}
L(t,r,\dot{r},\dot{\varphi})=\frac{m}{2}\dot{r}^{2}+\frac{mr^{2}}{2}\dot{%
\varphi}^{2}\underset{\overset{\shortparallel }{U_{s}(t,r)}}{\underbrace{%
-pr^{5}|V|e^{\frac{2|V|t}{r}}\cdot \dot{r}^{-1}+U(t,r)}},
\label{2D-Lagrangian}
\end{equation}%
where we have the following physical meanings: $\bullet $ $m$ is the \textit{%
mass} of the particule; $\bullet $ $V$ is the \textit{LB-monolayer
compressing rate}; $\bullet $ $p$ is a \textit{constant monolayer parameter}
given by the physical formula%
\begin{equation*}
p=\frac{\pi ^{2}q^{2}}{\varepsilon \varepsilon _{0}}\frac{\rho _{0}^{2}}{%
R_{0}^{2}};
\end{equation*}%
$\bullet $ $U_{s}(t,r)$ is an \textit{electrocapillarity potential energy}
including the monomolecular layer function%
\begin{eqnarray}
U(t,r) &=&p\left\{ \left[ -\frac{4}{3}r^{5}+\frac{16}{15}\left( |V|t\right)
r^{4}+\frac{1}{30}\left( |V|t\right) ^{2}r^{3}+\frac{1}{45}\left(
|V|t\right) ^{3}r^{2}+\right. \right.  \label{u} \\
&&\left. \left. +\frac{1}{45}\left( |V|t\right) ^{4}r+\frac{2}{45}\left(
|V|t\right) ^{5}\right] e^{\frac{2|V|t}{r}}-\frac{4}{45}\frac{\left(
|V|t\right) ^{6}}{r}f\left( \frac{2|V|t}{r}\right) \right\} .  \notag
\end{eqnarray}

The differential geometry (in the sense of nonlinear connections, Cartan
linear connections, d-torsions, d-curvatures etc.) produced by an arbitrary
jet Lagrangian function $L:J^{1}(\mathbb{R},M^{n})\rightarrow \mathbb{R}$ is
now completely done\ by Balan and Neagu in the monograph \cite{Balan-Neagu}.
The geometrical ideas from \cite{Balan-Neagu} are similar (but however
distinct ones) to those exposed by Miron and Anastasiei in the classical
Lagrangian geometry on tangent bundles (see \cite{Mir-An}). More accurately,
these jet geometrical Lagrangian ideas were initiated by Asanov in \cite%
{Asanov[2]} and developed further in multi-parameter way in the book \cite%
{Neagu Carte}, and in single-time way in \cite{Balan-Neagu}. In such a
context, this paper is devoted to the development on the $1$-jet vector
bundle $J^{1}(T,\mathbb{R}^{2})$ of the single-time Lagrange geometry (see 
\cite{Balan-Neagu}) produced by the 2D-monolayer Lagrangian (\ref%
{2D-Lagrangian}).

\section{The canonical nonlinear connection}

The local \textit{fundamental metrical d-tensor} produced by the
2D-monolayer Lagrangian (\ref{2D-Lagrangian}) is given by formula%
\begin{equation*}
g_{ij}=\frac{1}{2}\frac{\partial ^{2}L}{\partial y_{1}^{i}\partial y_{1}^{j}}%
.
\end{equation*}%
By direct computations, the fundamental metrical d-tensor $g_{ij}$ has the
matrix form%
\begin{equation}
g=\left( 
\begin{array}{cc}
g_{11} & g_{12} \\ 
g_{21} & g_{22}%
\end{array}%
\right) =\left( 
\begin{array}{cc}
\dfrac{m-2pr^{5}|V|e^{\frac{2|V|t}{r}}\cdot \dot{r}^{-3}}{2} & 0\medskip \\ 
0 & \dfrac{mr^{2}}{2}%
\end{array}%
\right) .  \label{g-jos-(ij)}
\end{equation}%
In order to have $\det g\neq 0$, we suppose that $g_{11}\neq 0$.
Consequently, the matrix $g=(g_{ij})$ admits the inverse $g^{-1}=(g^{jk})$,
whose entries are%
\begin{equation}
g^{-1}=\left( 
\begin{array}{cc}
g^{11} & g^{12} \\ 
g^{21} & g^{22}%
\end{array}%
\right) =\left( 
\begin{array}{cc}
\dfrac{2}{m-2pr^{5}|V|e^{\frac{2|V|t}{r}}\cdot \dot{r}^{-3}} & 0\medskip \\ 
0 & \dfrac{2}{mr^{2}}%
\end{array}%
\right) .  \label{g-sus-(jk)}
\end{equation}

Using a general formula from monograph \cite{Balan-Neagu}, we find the
following geometrical result:

\begin{proposition}
\label{semispray} For the 2D-monolayer Lagrangian (\ref{2D-Lagrangian}), the 
\textit{energy action functional}%
\begin{eqnarray*}
\mathbf{E}(t,r(t),\varphi (t)) &=&\int_{a}^{b}Ldt=\int_{a}^{b}\left[ \frac{m%
}{2}\left( \frac{dr}{dt}\right) ^{2}+\frac{mr^{2}}{2}\left( \frac{d\varphi }{%
dt}\right) ^{2}-pr^{5}|V|e^{\frac{2|V|t}{r}}\cdot \left( \frac{dr}{dt}%
\right) ^{-1}+U(t,r)\right] dt
\end{eqnarray*}%
produces on the $1$-jet space $J^{1}(T,\mathbb{R}^{2})$ the \textbf{%
canonical semispray} $\mathcal{S}=\left( H_{(1)1}^{(i)}=0,\text{ }%
G_{(1)1}^{(i)}\right) , $ where%
\begin{equation*}
\begin{array}{l}
G_{(1)1}^{(1)}=\dfrac{pr^{3}|V|e^{\frac{2|V|t}{r}}\left( 5r\dot{r}^{-1}-2|V|t%
\dot{r}^{-1}+|V|r\dot{r}^{-2}\right) -\dfrac{1}{2}\dfrac{\partial U}{%
\partial r}-\dfrac{mr}{2}\dot{\varphi}^{2}}{m-2pr^{5}|V|e^{\frac{2|V|t}{r}%
}\cdot \dot{r}^{-3}}\approx \medskip \\ 
\medskip \approx -\dfrac{1}{2}\dfrac{|V|}{r}\dot{r}+\left( \dfrac{|V|t}{r^{2}%
}-\dfrac{5}{2r}\right) \dot{r}^{2}-\dfrac{\dot{r}^{3}}{|V|}\left[ \dfrac{5}{3%
}r^{-1}-\dfrac{26\left( |V|t\right) }{15}r^{-2}+\dfrac{61\left( |V|t\right)
^{2}}{120}r^{-3}+\right. \\ 
\medskip \left. +\dfrac{\left( |V|t\right) ^{3}}{180}r^{-4}+\dfrac{\left(
|V|t\right) ^{4}}{180}r^{-5}+\dfrac{\left( |V|t\right) ^{5}}{90}r^{-6}-%
\dfrac{\left( |V|t\right) ^{6}}{45}r^{-7}e^{-\frac{2|V|t}{r}}f\left( \dfrac{%
2|V|t}{r}\right) \right] + \\ 
+\dfrac{m}{4p|V|}r^{-4}e^{-\frac{2|V|t}{r}}\dot{r}^{3}\dot{\varphi}%
^{2},\qquad G_{(1)1}^{(2)}=\dfrac{\dot{r}}{r}\dot{\varphi}.%
\end{array}%
\end{equation*}
\end{proposition}

\begin{proof}
The Euler-Lagrange equations of the energy action functional $\mathbf{E}$
are equivalent with the equations%
\begin{equation*}
\frac{d^{2}x^{i}}{dt^{2}}+2H_{(1)1}^{(i)}\left( t,x^{k},y_{1}^{k}\right)
+2G_{(1)1}^{(i)}\left( t,x^{k},y_{1}^{k}\right) =0,\qquad y_{1}^{k}=\frac{%
dx^{k}}{dt},
\end{equation*}%
where the local geometrical components%
\begin{equation*}
H_{(1)1}^{(i)}\overset{def}{=}-\frac{1}{2}\kappa _{11}^{1}(t)y_{1}^{i}=0
\end{equation*}%
and%
\begin{equation*}
\begin{array}{lll}
G_{(1)1}^{(i)} & \overset{def}{=} & \dfrac{g^{is}}{4}\left[ \dfrac{\partial
^{2}L}{\partial x^{q}\partial y_{1}^{s}}y_{1}^{q}-\dfrac{\partial L}{%
\partial x^{s}}+\dfrac{\partial ^{2}L}{\partial t\partial y_{1}^{s}}\right]%
\end{array}%
\end{equation*}%
represent a semispray on the $1$-jet space $J^{1}(T,\mathbb{R}^{2})$.
\end{proof}

The polynomial approximate form of the semispray $\mathcal{S}$ produces the
canonical nonlinear connection (for more details, see \cite{Balan-Neagu})%
\begin{equation*}
\Gamma =\left( M_{(1)1}^{(i)}=2H_{(1)1}^{(i)}=0,\text{ }N_{(1)j}^{(i)}=\frac{%
\partial G_{(1)1}^{(i)}}{\partial y_{1}^{j}}\right) .
\end{equation*}%
Consequently, we can enunciate the following geometrical result:

\begin{corollary}
\label{spatial nonlinear connection}The \textbf{canonical nonlinear
connection} produced by the 2D-monolayer Lagrangian (\ref{2D-Lagrangian})
has the following approximate spatial components:%
\begin{equation*}
\begin{array}{l}
\medskip N_{(1)1}^{(1)}=-\dfrac{1}{2}\dfrac{|V|}{r}+\left( \dfrac{2|V|t}{%
r^{2}}-\dfrac{5}{r}\right) \dot{r}-\mathcal{U}\left( t,r\right) \dot{r}^{2}+%
\dfrac{3me^{-\frac{2|V|t}{r}}}{4p|V|r^{4}}\dot{r}^{2}\dot{\varphi}^{2}, \\ 
N_{(1)2}^{(1)}=\dfrac{me^{-\frac{2|V|t}{r}}}{2p|V|r^{4}}\dot{r}^{3}\dot{%
\varphi},\qquad N_{(1)1}^{(2)}=\dfrac{\dot{\varphi}}{r},\qquad
N_{(1)2}^{(2)}=\dfrac{\dot{r}}{r},%
\end{array}%
\end{equation*}%
where%
\begin{eqnarray*}
\mathcal{U}\left( t,r\right) &=&\dfrac{1}{|V|}\left[ 5r^{-1}-\dfrac{26\left(
|V|t\right) }{5}r^{-2}+\dfrac{61\left( |V|t\right) ^{2}}{40}r^{-3}+\dfrac{%
\left( |V|t\right) ^{3}}{60}r^{-4}+\right. \\
&&\left. +\dfrac{\left( |V|t\right) ^{4}}{60}r^{-5}+\dfrac{\left(
|V|t\right) ^{5}}{30}r^{-6}-\dfrac{\left( |V|t\right) ^{6}}{15}r^{-7}e^{-%
\frac{2|V|t}{r}}f\left( \dfrac{2|V|t}{r}\right) \right] .
\end{eqnarray*}
\end{corollary}

\begin{proof}
Some direct computations and the formulas%
\begin{equation*}
N_{(1)1}^{(1)}=\frac{\partial G_{(1)1}^{(1)}}{\partial \dot{r}},\quad
N_{(1)2}^{(1)}=\frac{\partial G_{(1)1}^{(1)}}{\partial \dot{\varphi}},\quad
N_{(1)1}^{(2)}=\frac{\partial G_{(1)1}^{(2)}}{\partial \dot{r}},\quad
N_{(1)2}^{(2)}=\frac{\partial G_{(1)1}^{(2)}}{\partial \dot{\varphi}}
\end{equation*}%
imply the required result.
\end{proof}

\section{The Cartan canonical $\Gamma $-linear connection and its d-torsions}

We use the nonlinear connection from Corollary \ref{spatial nonlinear
connection} to construct the dual \textit{adapted bases} of d-vector fields%
\begin{equation}
\left\{ \frac{\partial }{\partial t}\text{ };\text{ }\frac{\delta }{\delta r}%
\text{ };\text{ }\frac{\delta }{\delta \varphi }\text{ };\text{ }\dfrac{%
\partial }{\partial \dot{r}}\text{ };\text{ }\dfrac{\partial }{\partial \dot{%
\varphi}}\right\} \subset \mathcal{X}(E)  \label{a-b-v}
\end{equation}%
and d-covector fields%
\begin{equation}
\left\{ dt\text{ };\text{ }dr\text{ };\text{ }d\varphi \text{ };\text{ }%
\delta \dot{r}\text{ };\text{ }\delta \dot{\varphi}\right\} \subset \mathcal{%
X}^{\ast }(E),  \label{a-b-co}
\end{equation}%
where $E=J^{1}(T,\mathbb{R}^{2})$, and we set%
\begin{equation*}
\begin{array}{l}
\dfrac{\delta }{\delta r}=\dfrac{\partial }{\partial r}-N_{(1)1}^{(1)}\dfrac{%
\partial }{\partial \dot{r}}-\dfrac{\dot{\varphi}}{r}\dfrac{\partial }{%
\partial \dot{\varphi}},\qquad\dfrac{\delta }{\delta \varphi }=\dfrac{%
\partial }{\partial \varphi }-\dfrac{me^{-\frac{2|V|t}{r}}}{2p|V|r^{4}}\dot{r%
}^{3}\dot{\varphi}\dfrac{\partial }{\partial \dot{r}}-\dfrac{\dot{r}}{r}%
\dfrac{\partial }{\partial \dot{\varphi}},\medskip \\ 
\delta \dot{r}=d\dot{r}+N_{(1)1}^{(1)}dr+\dfrac{me^{-\frac{2|V|t}{r}}}{%
2p|V|r^{4}}\dot{r}^{3}\dot{\varphi}d\varphi ,\qquad\delta \dot{\varphi}=d%
\dot{\varphi}+\dfrac{\dot{\varphi}}{r}dr+\dfrac{\dot{r}}{r}d\varphi .%
\end{array}%
\end{equation*}%
Note that, under a change of coordinates (\ref{tr-rules-2}), the elements of
the adapted bases (\ref{a-b-v}) and (\ref{a-b-co}) transform as classical
tensors. Consequently, all subsequent geometrical objects on the $1$-jet
space $J^{1}(T,\mathbb{R}^{2})$ (such as the Cartan canonical linear
connection, its torsion etc.) will be described in local adapted components.

Using a general result from \cite{Balan-Neagu}, by direct computations, we
give the following important geometrical result:

\begin{proposition}
The \textbf{Cartan canonical }$\Gamma $\textbf{-linear connection} of the
2D-monolayer Lagrangian (\ref{2D-Lagrangian}) has the following approximate
adapted components:%
\begin{equation*}
C\Gamma =\left( \kappa _{11}^{1}=0,\text{ }G_{j1}^{k},\text{ }L_{jk}^{i},%
\text{ }C_{j(k)}^{i(1)}\right) ,
\end{equation*}%
where%
\begin{equation*}
G_{j1}^{k}=\delta _{1}^{k}\cdot \delta _{j1}\cdot \dfrac{2pr^{4}|V|^{2}}{%
2pr^{5}|V|-m\dot{r}^{3}e^{-\frac{2|V|t}{r}}},
\end{equation*}%
\begin{equation*}
C_{j(k)}^{i(1)}=\delta _{1}^{i}\cdot \delta _{j1}\cdot \delta _{k1}\cdot 
\dfrac{3pr^{5}|V|}{m\dot{r}^{4}e^{-\frac{2|V|t}{r}}-2pr^{5}|V|\dot{r}}
\end{equation*}%
and%
\begin{equation*}
L_{11}^{1}=\dfrac{pr^{3}|V|\left( 2|V|t-5r\right) }{m\dot{r}^{3}e^{-\frac{%
2|V|t}{r}}-2pr^{5}|V|}-N_{(1)1}^{(1)}C_{1(1)}^{1(1)},
\end{equation*}%
\begin{equation*}
L_{12}^{1}=L_{21}^{1}=-\dfrac{me^{-\frac{2|V|t}{r}}}{2p|V|r^{4}}\dot{r}^{3}%
\dot{\varphi}\cdot C_{1(1)}^{1(1)},\qquad L_{22}^{1}=\dfrac{mr}{2pr^{5}|V|e^{%
\frac{2|V|t}{r}}\dot{r}^{-3}-m},
\end{equation*}%
\begin{equation*}
L_{11}^{2}=\dfrac{3\dot{\varphi}}{2r\dot{r}},\quad L_{12}^{2}=L_{21}^{2}=%
\dfrac{1}{r},\quad L_{22}^{2}=0.
\end{equation*}
\end{proposition}

\begin{proof}
Via the 2D-monolayer adapted derivative operators (\ref{a-b-v}), we use the
general formulas which give the adapted components of a Cartan canonical
connection (see \cite{Balan-Neagu}):%
\begin{equation*}
G_{j1}^{k}=\frac{g^{ks}}{2}\frac{\partial g_{sj}}{\partial t},\qquad
C_{j(k)}^{i(1)}=\frac{g^{is}}{2}\left( \frac{\partial g_{js}}{\partial
y_{1}^{k}}+\frac{\partial g_{ks}}{\partial y_{1}^{j}}-\frac{\partial g_{jk}}{%
\partial y_{1}^{s}}\right) =\frac{g^{is}}{2}\frac{\partial g_{js}}{\partial
y_{1}^{k}},
\end{equation*}%
\begin{eqnarray*}
L_{jk}^{i} &=&\frac{g^{is}}{2}\left( \frac{\delta g_{js}}{\delta x^{k}}+%
\frac{\delta g_{ks}}{\delta x^{j}}-\frac{\delta g_{jk}}{\delta x^{s}}\right)
= \\
&=&\Gamma
_{jk}^{i}-N_{(1)k}^{(q)}C_{j(q)}^{i(1)}-N_{(1)j}^{(q)}C_{k(q)}^{i(1)}+\frac{%
g^{is}}{2}N_{(1)s}^{(q)}\frac{\partial g_{jk}}{\partial y_{1}^{q}}= \\
&=&\Gamma
_{jk}^{i}-N_{(1)k}^{(1)}C_{j(1)}^{i(1)}-N_{(1)j}^{(1)}C_{k(1)}^{i(1)}+\frac{%
g^{is}}{2}N_{(1)s}^{(q)}\frac{\partial g_{jk}}{\partial y_{1}^{q}},
\end{eqnarray*}%
where%
\begin{equation*}
\Gamma _{jk}^{i}=\frac{g^{is}}{2}\left( \frac{\partial g_{js}}{\partial x^{k}%
}+\frac{\partial g_{ks}}{\partial x^{j}}-\frac{\partial g_{jk}}{\partial
x^{s}}\right) ,
\end{equation*}%
and we remind that we have $(t,x^{1},x^{2},y_{1}^{1},y_{1}^{2})=\left(
t,r,\varphi ,\dot{r},\dot{\varphi}\right) $.
\end{proof}

\begin{remark}
The Cartan canonical connection $C\Gamma $ has the me\-tri\-cal properties%
\begin{equation*}
\begin{array}{lll}
g_{ij/1}=g_{\text{ \ }/1}^{ij}=0, & g_{ij|k}=g_{\text{ \ }|k}^{ij}=0, & 
g_{ij}|_{(k)}^{(1)}=g^{ij}|_{(k)}^{(1)}=0,%
\end{array}%
\end{equation*}%
where \textquotedblright $_{/1}$\textquotedblright , \textquotedblright $%
_{|k}$\textquotedblright\ and \textquotedblright $|_{(k)}^{(1)}$%
\textquotedblright\ are the $T$\textbf{-horizontal}\textit{, }$\mathbb{R}%
^{2} $\textbf{-horizontal }and \textbf{vertical covariant derivatives }%
produced by the Cartan linear connection $C\Gamma $ (for more details, see 
\cite{Balan-Neagu}). Consequently, in our jet single-time Lagrange
geometrization, the Cartan canonical connection plays a similar role to that
of the Levi-Civita connection from classical Riemannian geometry.
\end{remark}

\begin{proposition}
The Cartan canonical connection $C\Gamma $ of the 2D-monolayer Lagrangian (%
\ref{2D-Lagrangian}) has the next approximate adapted local \textbf{torsion
d-tensors}:%
\begin{equation*}
T_{1j}^{k}=\mathcal{P}_{(1)1(j)}^{(k)\text{ }(1)}=-G_{j1}^{k},\qquad
P_{i(j)}^{k(1)}=C_{i(j)}^{k(1)},
\end{equation*}%
\begin{equation*}
H_{(1)1j}^{(k)}=\delta _{1}^{k}\cdot \left[ \delta _{j1}\cdot \left( \frac{%
\partial \mathcal{U}}{\partial t}\dot{r}^{2}-\frac{2|V|}{r^{2}}\dot{r}+%
\dfrac{3me^{-\frac{2|V|t}{r}}}{2pr^{5}}\dot{r}^{2}\dot{\varphi}^{2}\right)
+\delta _{j2}\cdot \dfrac{me^{-\frac{2|V|t}{r}}}{pr^{5}}\dot{r}^{3}\dot{%
\varphi}\right] ,
\end{equation*}%
\begin{equation*}
R_{(1)12}^{(1)}=-R_{(1)21}^{(1)}=\dfrac{me^{-\frac{2|V|t}{r}}}{p|V|r^{4}}%
\left[ \frac{3}{2}N_{(1)1}^{(1)}-\frac{1}{2}\frac{\partial N_{(1)1}^{(1)}}{%
\partial \dot{r}}\dot{r}+\left( \frac{1}{r}-\dfrac{|V|t}{r^{2}}\right) \dot{r%
}\right] \dot{r}^{2}\dot{\varphi},
\end{equation*}%
\begin{equation*}
R_{(1)12}^{(2)}=-R_{(1)21}^{(2)}=\frac{1}{r}N_{(1)1}^{(1)},
\end{equation*}%
\begin{equation*}
R_{(1)11}^{(1)}=R_{(1)22}^{(1)}=R_{(1)11}^{(2)}=R_{(1)22}^{(2)}=0,
\end{equation*}%
\begin{equation*}
P_{(1)1(1)}^{(1)\text{ }(1)}=\frac{\partial N_{(1)1}^{(1)}}{\partial \dot{r}}%
+N_{(1)1}^{(1)}C_{1(1)}^{1(1)}-\dfrac{pr^{3}|V|\left( 2|V|t-5r\right) }{m%
\dot{r}^{3}e^{-\frac{2|V|t}{r}}-2pr^{5}|V|},
\end{equation*}%
\begin{equation*}
P_{(1)1(2)}^{(1)\text{ }(1)}=P_{(1)2(1)}^{(1)\text{ }(1)}=\dfrac{me^{-\frac{%
2|V|t}{r}}}{2p|V|r^{4}}\left( 3+\dot{r}\cdot C_{1(1)}^{1(1)}\right) \dot{r}%
^{2}\dot{\varphi},
\end{equation*}%
\begin{equation*}
P_{(1)2(2)}^{(1)\text{ }(1)}=\dfrac{me^{-\frac{2|V|t}{r}}}{2p|V|r^{4}\dot{r}%
^{-3}}-\dfrac{mr}{2pr^{5}|V|e^{\frac{2|V|t}{r}}\dot{r}^{-3}-m},
\end{equation*}%
\begin{equation*}
P_{(1)1(1)}^{(2)\text{ }(1)}=-\dfrac{3\dot{\varphi}}{2r\dot{r}},\qquad
P_{(1)1(2)}^{(2)\text{ }(1)}=P_{(1)2(1)}^{(2)\text{ }(1)}=P_{(1)2(2)}^{(2)%
\text{ }(1)}=0.
\end{equation*}
\end{proposition}

\begin{proof}
Generally, a Cartan canonical connection on the $1$-jet space $J^{1}(T,%
\mathbb{R}^{2})$ is characterized by \textit{six} effective d-tensors of
torsion (for more details, see \cite{Balan-Neagu}). For our Cartan canonical
connection $C\Gamma $ these reduce to the following expressions:%
\begin{equation*}
\begin{array}{ccc}
T_{1j}^{k}=-G_{j1}^{k}, & H_{(1)1j}^{(k)}=-{\dfrac{\partial N_{(1)j}^{(k)}}{%
\partial t}}, & R_{(1)ij}^{(k)}={\dfrac{\delta N_{(1)i}^{(k)}}{\delta x^{j}}}%
-{\dfrac{\delta N_{(1)j}^{(k)}}{\delta x^{i}}},\medskip \\ 
\mathcal{P}_{(1)1(j)}^{(k)\text{ }(1)}=-G_{j1}^{k}, & 
P_{i(j)}^{k(1)}=C_{i(j)}^{k(1)}, & {P_{(1)i(j)}^{(k)\text{ }(1)}={\dfrac{%
\partial N_{(1)i}^{(k)}}{\partial y_{1}^{j}}}-L_{ij}^{k},}%
\end{array}%
\end{equation*}%
where we have $(t,x^{1},x^{2},y_{1}^{1},y_{1}^{2})=\left( t,r,\varphi ,\dot{r%
},\dot{\varphi}\right) $.
\end{proof}

\section{From the dynamics of 2D-monolayer to a jet geometrical Yang-Mills
energy\label{Y-M-energy}}

In monograph \cite{Balan-Neagu}, using only a given Lagrangian function $%
L(t,x,y)$ on the $1$-jet space $J^{1}(\mathbb{R},M^{n})$, an
electromagnetic-like geometrical model was also created. In the background
of the electromagnetic-like geometrical formalism from \cite{Balan-Neagu},
one works with an \textit{electromagnetic d-form} (this Section the Latin
letters run from $1$ to $n$) $\mathbb{F}=F_{(i)j}^{(1)}\delta
y_{1}^{i}\wedge dx^{j},$ where%
\begin{equation*}
F_{(i)j}^{(1)}=\frac{h^{11}}{2}\left[
g_{js}N_{(1)i}^{(s)}-g_{is}N_{(1)j}^{(s)}+\left(
g_{iq}L_{js}^{q}-g_{jq}L_{is}^{q}\right) y_{1}^{s}\right] .
\end{equation*}%
The above electromagnetic components are characterized by the following 
\textit{geometrical Maxwell-like equations} (for more details, see \cite%
{Balan-Neagu}):%
\begin{eqnarray*}
F_{(i)j/1}^{(1)} &=&\frac{1}{2}\mathcal{A}_{\left\{ i,j\right\} }\left\{ 
\overline{D}_{(i)1|j}^{(1)}-D_{(i)s}^{(1)}G_{j1}^{s}+d_{(i)(s)}^{(1)(1)}%
\left( {\dfrac{\delta M_{(1)1}^{(s)}}{\delta x^{j}}}-{\dfrac{\delta
N_{(1)j}^{(s)}}{\delta t}}\right) -\right. \\
&&\left. -\left[ C_{j(s)}^{l(1)}\left( {\dfrac{\delta M_{(1)1}^{(s)}}{\delta
x^{i}}}-{\dfrac{\delta N_{(1)i}^{(s)}}{\delta t}}\right) -G_{i1|j}^{l}\right]
h^{11}g_{lq}y_{1}^{q}\right\} ,
\end{eqnarray*}%
\begin{equation*}
\sum_{\{i,j,k\}}F_{(i)j|k}^{(1)}=-\frac{1}{4}\sum_{\{i,j,k\}}\frac{\partial
^{3}L}{\partial y_{1}^{i}\partial y_{1}^{l}\partial y_{1}^{s}}\left[ {\dfrac{%
\delta N_{(1)j}^{(s)}}{\delta x^{k}}}-{\dfrac{\delta N_{(1)k}^{(s)}}{\delta
x^{j}}}\right] y_{1}^{l},
\end{equation*}%
\begin{equation*}
\sum_{\{i,j,k\}}F_{(i)j}^{(1)}|_{(k)}^{(1)}=0,
\end{equation*}%
where $\mathcal{A}_{\left\{ i,j\right\} }$ means an alternate sum, $%
\sum_{\{i,j,k\}}$ means a cyclic sum, and we have%
\begin{equation*}
\overline{D}_{(i)1}^{(1)}=\frac{h^{11}}{2}\left( \frac{\partial g_{is}}{%
\partial t}-M_{(1)1}^{(q)}\frac{\partial g_{is}}{\partial y_{1}^{q}}\right)
y_{1}^{s},\quad D_{(i)j}^{(1)}=h^{11}g_{iq}\left[
-N_{(1)j}^{(q)}+L_{js}^{q}y_{1}^{s}\right] ,
\end{equation*}%
\begin{equation*}
d_{(i)(j)}^{(1)(1)}=h^{11}\left[ g_{ij}+g_{iq}C_{s(j)}^{q(1)}y_{1}^{s}\right]
,
\end{equation*}%
\begin{equation*}
\overline{D}_{(i)1|j}^{(1)}=\frac{\delta \overline{D}_{(i)1}^{(1)}}{\delta
x^{j}}-\overline{D}_{(s)1}^{(1)}L_{ij}^{s},\quad G_{i1|j}^{k}=\frac{\delta
G_{i1}^{k}}{\delta x^{j}}+G_{i1}^{s}L_{sj}^{k}-G_{s1}^{k}L_{ij}^{s},
\end{equation*}%
\begin{equation*}
F_{(i)j/1}^{(1)}=\frac{\partial F_{(i)j}^{(1)}}{\partial t}-M_{(1)1}^{(q)}%
\frac{\partial F_{(i)j}^{(1)}}{\partial y_{1}^{q}}+F_{(i)j}^{(1)}\kappa
_{11}^{1}-F_{(s)j}^{(1)}G_{i1}^{s}-F_{(i)s}^{(1)}G_{j1}^{s},
\end{equation*}%
\begin{equation*}
F_{(i)j|k}^{(1)}=\frac{\delta F_{(i)j}^{(1)}}{\delta x^{k}}%
-F_{(s)j}^{(1)}L_{ik}^{s}-F_{(i)s}^{(1)}L_{jk}^{s},\qquad
F_{(i)j}^{(1)}|_{(k)}^{(1)}=\frac{\partial F_{(i)j}^{(1)}}{\partial y_{1}^{k}%
}-F_{(s)j}^{(1)}C_{i(k)}^{s(1)}-F_{(i)s}^{(1)}C_{j(k)}^{s(1)}.
\end{equation*}

\begin{example}
The Lagrangian function that governs the movement law of a particle of mass $%
m\neq 0$ and electric charge $e$, which is concomitantly displaced into an
environment endowed both with a gravitational field and an electromagnetic
one, is given by (see \cite{Mir-An} and \cite{Balan-Neagu})%
\begin{equation}
L(t,x^{k},y_{1}^{k})=mch^{11}(t)\varphi _{ij}(x^{k})y_{1}^{i}y_{1}^{j}+{%
\frac{2e}{m}}A_{(i)}^{(1)}(x^{k})y_{1}^{i}+\mathcal{F}(t,x^{k}),
\label{x05-L-(jet)-ED}
\end{equation}%
where the semi-Riemannian metric $\varphi _{ij}(x)$ represents the \textbf{%
gravitational potentials}, $A_{(i)}^{(1)}(x)$ represent the \textbf{%
electromagnetic potential}\emph{, }and $\mathcal{F}(t,x)$ is a smooth 
\textbf{potential function}. In the case of the usual time $(\mathbb{R}%
,h_{11}(t)=1),$ applying to the Lagrangian (\ref{x05-L-(jet)-ED}) our jet
geometrical formalism$,$ the \textbf{electromagnetic components} take the
form%
\begin{equation*}
F_{(i)j}^{(1)}=-\frac{e}{2m}\left( \frac{\partial A_{(i)}^{(1)}}{\partial
x^{j}}-\frac{\partial A_{(j)}^{(1)}}{\partial x^{i}}\right) ,
\end{equation*}%
and the \textbf{Maxwell geometrical equations}\ reduce to the classical
ones: $\sum_{\{i,j,k\}}F_{(i)j|k}^{(1)}=0, $ where ($\gamma _{jk}^{i}$ are
the Christoffel symbols of the semi-Riemannian metric $\varphi _{ij}$)%
\begin{equation*}
F_{(i)j|k}^{(1)}=\frac{\partial F_{(i)j}^{(1)}}{\partial x^{k}}%
-F_{(m)j}^{(1)}\gamma _{ik}^{m}-F_{(i)m}^{(1)}\gamma _{jk}^{m}.
\end{equation*}%
In our opinion, this fact suggests some kind of naturalness attached to our
jet single-time Lagrangian geometrical electromagnetic-like theory.
\end{example}

On our particular $1$-jet space $J^{1}(T,\mathbb{R}^{2})$, the 2D-monolayer
Lagrangian (\ref{2D-Lagrangian}) produces the electromagnetic-like $2$-form 
\begin{equation*}
\mathbb{F}=F_{(1)1}^{(1)}\delta \dot{r}\wedge dr+F_{(1)2}^{(1)}\delta \dot{r}%
\wedge d\varphi +F_{(2)1}^{(1)}\delta \dot{\varphi}\wedge
dr+F_{(2)2}^{(1)}\delta \dot{\varphi}\wedge d\varphi ,
\end{equation*}%
where $F_{(1)1}^{(1)}=F_{(2)2}^{(1)}=0$ and%
\begin{equation*}
F_{(2)1}^{(1)}=-F_{(1)2}^{(1)}=\dfrac{1}{2}\left[ \dfrac{3mr}{2}+\dfrac{%
m^{2}e^{-\frac{2|V|t}{r}}}{4p|V|r^{4}}\cdot \dot{r}^{3}\right] \cdot \dot{%
\varphi}.
\end{equation*}

The magnitude (i.e., the \textit{jet geometrical Yang-Mills energy}) of our
theoretical electromagnetic-like field $\mathbb{F}$ produced by the
2D-monolayer Lagrangian (\ref{2D-Lagrangian}) is defined by the formula%
\begin{equation*}
\mathcal{EYM}(\mathbb{F)=}\frac{1}{2m}\cdot \text{Trace}\left[ \mathbf{F}%
\cdot \text{ }^{T}\mathbf{F}\right] =\frac{1}{m}\cdot \left[ F_{(1)2}^{(1)}%
\right] ^{2},
\end{equation*}%
where the skew-symmetric electromagnetic matrix is%
\begin{equation*}
\mathbf{F}=\left( 
\begin{array}{cc}
0 & F_{(1)2}^{(1)} \\ 
-F_{(1)2}^{(1)} & 0%
\end{array}%
\right) \in o(2)=L(O(2)).
\end{equation*}%
Note that the matrix $\mathbf{F}$ belongs to the Lie algebra $o(2)$ of the
Lie group of orthogonal matrices%
\begin{equation*}
O(2)=\left\{ A\in GL_{2}(\mathbb{R})\text{ }|\text{ }A\cdot \text{ }%
^{T}A=I_{2}\right\} .
\end{equation*}

The jet geometrical Yang-Mills energy produced by the 2D-monolayer
Lagrangian (\ref{2D-Lagrangian}) cancels iff%
\begin{equation*}
F_{(1)2}^{(1)}=0\Leftrightarrow \left[ \dfrac{3mr}{2}+\dfrac{m^{2}e^{-\frac{%
2|V|t}{r}}}{4p|V|r^{4}}\cdot \dot{r}^{3}\right] \cdot \dot{\varphi}=0.
\end{equation*}

\section{Jacobi equations for deviation of geodesics from an instanton-like
solution. Computer imaging simulations and physical interpretations}

%Рассмотрим сингулярные траектории, которые дает
Let us find the singular trajectories of the differential equations 
%уравнение из
from the Proof of Proposition \ref{semispray}: 
\begin{equation}
{\dfrac{dy_{1}^{i}}{dt}}+2G_{(1)1}^{(i)}\left( t,x^{k},y_{1}^{k}\right)
=0,\qquad \dfrac{dx^{k}}{dt}=y_{1}^{k},  \label{geodesic-equation}
\end{equation}%
with 
\begin{equation}
G_{(1)1}^{(i)}=\dfrac{g^{is}}{4}\left[ \dfrac{\partial ^{2}L}{\partial
x^{q}\partial y_{1}^{s}}y_{1}^{q}-\dfrac{\partial L}{\partial x^{s}}+\frac{%
\partial ^{2}L}{\partial t\partial y_{1}^{s}}\right] ,
\label{add1to-geodesic-equation}
\end{equation}%
where%
\begin{equation}
\begin{array}{l}
G_{(1)1}^{(1)}=\dfrac{pr^{3}|V|e^{\frac{2|V|t}{r}}\left( 5r\dot{r}^{-1}-2|V|t%
\dot{r}^{-1}+|V|r\dot{r}^{-2}\right) -\dfrac{1}{2}\dfrac{\partial U}{%
\partial r}-\dfrac{mr}{2}\dot{\varphi}^{2}}{m-2pr^{5}|V|e^{\frac{2|V|t}{r}%
}\cdot \dot{r}^{-3}},\qquad G_{(1)1}^{(2)}=\dfrac{\dot{r}}{r}\dot{\varphi}.%
\end{array}
\label{add2to-geodesic-equation}
\end{equation}%
One can note that a singularity in (\ref{add2to-geodesic-equation})
represents itself a zero difference of the kinetic energy and a part of the
potential energy: 
%Замечаем, что сингулярность имеет вид нулевой разности кинетической энергии
%и части потенциальной энергии:
\begin{equation}
\frac{m}{2}\dot{r}^{2}+\frac{mr^{2}}{2}\dot{\varphi}^{2}-pr^{5}|V|e^{\frac{%
2|V|t}{r}}\cdot \dot{r}^{-1},\qquad \dot{\varphi}\rightarrow 0.
\label{singular-differenc-2D-Lagrangian}
\end{equation}%
Therefore, in the neighbourhood of singular behavior 
%Поэтому в окрестности сингулярного поведения
of $G_{(1)1}^{(1)}$, the effects stipulated by a rest % эффекты,
%обусловленные остатком
$U$ of the potential energy %от потенциальной энергии
$U_{s}$ are small. %, малы.
Consequently, these above allow us to make the proposal that the Lagrangian 
% Вышесказанное позволяет сделать предположение о возможности
%разложения лагранжиана
$L$ can be expanded into series %в ряд
$L=L_{0}+\delta L,$ where the unperturbed Lagrangian 
%где невозмущенный лагранжиан
$L_{0}$ gives trajectories being nearby the trajectories with zero
Hamiltonian: %дает траектории, близкие к
%траектория с нулевым значением гамильтониана:
\begin{equation}
E_{inst}=\frac{m}{2}\dot{r}^{2}+\frac{mr^{2}}{2}\dot{\varphi}%
^{2}+pr^{5}|V|e^{\frac{2|V|t}{r}}\cdot \dot{r}^{-1}-U=0.
\label{2D-Lagrangian-singular-part}
\end{equation}

A solution which obeys condition % Решение, подчиняющееся условию
(\ref{2D-Lagrangian-singular-part}) is called %,
%называется
an \textit{instanton-like solution}. In our case, the geodesic trajectory of
a particle in % геодезическая траектория
2D-membrane is closed %мембраны  близка
to the instanton-like solution $\vec{r_{0}}$. It has been shown in \cite%
{Grushevskaya2011}, \cite{Grushevskaya2012} that a calculation of the
trajectories 
%  было показано, что расчет  траекторий, энергия которых близка к
having zero energy in a field of the electrocapillarity forces of monolayer 
% нулевой в поле электрокапилярных сил монослоя, может быть проведен в
can be performed in a resonant approximation 
%так называемом резонансном приближении в виде поправки
in form of a correction $\vec{\delta r}$ for $\vec{r_{0}}$. The correction 
%Поправка
$\vec{\delta r}$ is described by some Jacobi equations 
% описывается уравнением Якоби,
which are equivalent to some equations that describe the physical system
under the action 
% эквивалентным уравнению, которое описывает физическую систему под
%воздействием
of the field having the frequency closed to 
%  поля, имеющего частоту близкую к собственной частоте системы. Такое
%воздействие называется резонансным.
an eigenfrequency of the system. Such type of action is called a \textit{%
resonant action}. Now, let us simulate\footnote{%
Please note that the four computer graphic simulations used and cited in
this preliminary version of our paper are missing because of some uploading
problems. If some researcher is interested in these graphics, we warmly and
cordially invite him to contact any author of this research work.} the
trajectories described by the equations 
%численно найдем траектории, описываемые уравнениями
(\ref{geodesic-equation} -- \ref{add2to-geodesic-equation}). These
trajectories are represented %Эти траектории представлены
in Fig.~1. According to images in the figure %Из этих рисунков
%следует, что
steady states are limiting cycles, and a relaxation to the limiting cycle is
radial one with $\dot{\varphi}\rightarrow 0$ for the large speeds. Hence,
the electrocapillar action on the system under consideration holds a
spherical symmetry.

Therefore, as resonant action on a particle of the monolayer one chooses a
geometrical Yang-Mills field with zero energy and 
% Поэтому выберем поле Янга -- Миллса с нулевой энергией и
the symmetry group $O(2)$, that was proposed in Section %которое
%предложено в параграфе
\ref{Y-M-energy}: 
\begin{equation}
F_{(1)2}^{(1)}\overset{def}{=}\left[ \dfrac{3mr}{2}+\dfrac{m^{2}e^{\frac{%
-2|V|t}{r}}}{4p|V|r^{4}}\cdot \dot{r}^{3}\right] =0.  \label{Yang-mils-field}
\end{equation}%
%
%
%
%
%
%
%
%
%
%
%
%
%
%
%
%
%
%
%
%
%
%
%в качестве резонансного воздействия на частицу монослоя.

Now, we can find the trajectory of a particle in the field 
%Найдем траекторию частицы в поле
(\ref{Yang-mils-field}). To do that, it is necessary to make a change 
%Для этого необходимо сделать замену
$t\rightarrow -t,\ \dot{r}\rightarrow -\dot{r}$ because the centrifugal
acceleration 
%, поскольку центробежное ускорение частице монослоя противоположно по
of the particle of monolayer has opposite direction with respect to the
centripetal acceleration produced by this field. 
%направлению центростремительному ускорению, создаваемому этим полем.
After this change one gets that the components %компоненты
$F_{(1)2}^{(1)}(-t)$ are proportional to %, которые пропорциональны to
\begin{equation}
\left[ \dfrac{3mr_{0}}{2}-\dfrac{m^{2}e^{\frac{2|V|t}{r_{0}}}}{4p|V|r_{0}^{4}%
}\cdot \dot{r_{0}}^{3}\right] =0,
\end{equation}%
and the corresponding trajectories are: %соответствующие траектории:
\begin{equation}
\dfrac{3r_{0}}{2}-\dfrac{me^{\frac{2|V|t}{r_{0}}}}{4p|V|r_{0}^{4}}\cdot \dot{%
r_{0}}^{3}=0\qquad \mbox{for arbitrary}\quad \dot{\varphi _{0}}.
\label{resonance-trajectory}
\end{equation}

Comparing the expressions %Сравнивая выражения
(\ref{singular-differenc-2D-Lagrangian}) and %и
(\ref{resonance-trajectory}), one can conclude that the trajectory of a
monolayer particle in the geometrical Yang-Mills field 
% между собой, заключаем, что траектория
%частицы монослоя  в поле
%Янга -- Милса
(\ref{Yang-mils-field}) with zero energy %с нулевой энергией
%близка
is closed to the solution of Eqs.~(\ref{geodesic-equation}, \ref%
{add1to-geodesic-equation}, \ref{add2to-geodesic-equation}). Hence, the
physical meaning of the introduced geometrical field 
%Следовательно, физический
%смысл введенного поля
(\ref{Yang-mils-field}) is that this field is a field of the
electrocapillarity forces which acts resonantly on monolayer molecules. 
% заключается в том, что это -- поле
% электрокапиллярных сил, резонансно воздействующее на молекулы монослоя.

Generally, the instanton-like solution is an unstable state of perturbation,
which is a transition from one to at least two stable states. %
%  является неустойчивым
%состоянием, возмущение которого является переходом
%в одно из, по крайней мере двух устойчивых состояний.
In our case, the perturbation describes %описывает
a first-order phase transition, %фазовый переход. And мы будем понимать
and we will mean this transition phase as an exit of the system from the
resonance with the Yang-Mills field. 
%под фазовым переходом выход системы из резонанса с полем Янга -- Милса.
Further we will describe a resonance detuning. 
% Далее опишем рассогласование резонанса
This process is a deviation\footnote{%
Please note that here $\delta $ means a small variation. This does not
represent the covariant derivative $\delta $ used in the previous
geometrical construction of distinguished $1$-forms $\delta \dot{r}$ and $%
\delta \dot{\varphi}$.\nopagebreak} %отклонение
$\vec{\delta r}=\{\delta r(t),\delta \varphi (t)\}$ of %от следования
a solution of Eq.~(\ref{geodesic-equation}) along the trajectory 
%вдоль траектории
$\vec{r}_{0}=\{\delta r_{0},\delta \varphi _{0}\}$ (\ref%
{resonance-trajectory}) of the zero Yang-Mills field 
%нулевого поля Янга -- Милса
(\ref{Yang-mils-field}): 
\begin{equation}
\vec{r}(t)=\vec{r}_{0}(-t)+\vec{\delta r}(t).  \label{perturb-trajectory}
\end{equation}

Now, let us find a Hamiltonian of the system under above considerations,
taking into account that 
\begin{equation*}
L=\frac{1}{2}m{\dot{r}}^{2}+U(r,t)+\frac{1}{2}mr^{2}\dot{\varphi}^{2}-\frac{%
p|V|r^{5}e^{\frac{2t|V|}{r}}}{\dot{r}},
\end{equation*}%
where $U(r,t)$ is given by (\ref{u}), and the non-zero components of the
metric tensor are % компоненты метрического тензора, не равные нулю, есть
\begin{equation*}
g_{11}=\frac{1}{2}\left( m-\frac{2p|V|r^{5}e^{\frac{2tV}{r}}}{{\dot{r}}^{3}}%
\right) ,\qquad g_{22}=\frac{1}{2}mr^{2}.
\end{equation*}

Let us use instead of classical Hamiltonian the difference between $g_{11}{%
\dot{r}}^{2}+g_{22}\dot{\varphi}^{2}$ and the given Lagrangian. It follows
that our renormalized Hamiltonian reads 
\begin{equation}
\begin{array}{l}
H=g_{11}{\dot{r}}^{2}+g_{22}\dot{\varphi}^{2}-L=\medskip \\ 
=-p\left[ e^{\frac{2t|V|}{r}}\left( \dfrac{1}{45}t^{4}|V|^{4}r+\dfrac{1}{45}%
t^{3}|V|^{3}r^{2}+\dfrac{1}{30}t^{2}|V|^{2}r^{3}+\dfrac{16}{15}t|V|r^{4}-%
\dfrac{4r^{5}}{3}+\dfrac{2t^{5}|V|^{5}}{45}\right) -\right. \medskip \\ 
\left. -\dfrac{4t^{6}|V|^{6}}{45r}f\left( \dfrac{2t|V|}{r}\right) \right] +%
\dfrac{1}{2}{\dot{r}}^{2}\left( m-\dfrac{2p|V|r^{5}e^{\frac{2t|V|}{r}}}{{%
\dot{r}}^{3}}\right) -\dfrac{1}{2}m{\dot{r}}^{2}+\dfrac{p|V|r^{5}e^{\frac{%
2t|V|}{r}}}{\dot{r}}.%
\end{array}
\notag
\end{equation}

Let us consider now that the unperturbed renormalized Hamiltonian is the
Yang-Mills energy %энергией Янга -- Милса
\begin{equation*}
H_{YM}=\frac{\dot{\varphi}^{2}\left( \frac{m^{2}e^{-\frac{2t|V|}{r}}{\dot{r}}%
^{3}}{4p|V|r^{4}}+\frac{3}{2}mr\right) ^{2}}{4m}.
\end{equation*}%
Then, the perturbation is determined as 
\begin{equation}
\delta L=-\Delta H=-(H-H_{YM})=-\left[ -U(r,t)-\frac{m\dot{\varphi}^{2}e^{-%
\frac{4t|V|}{r}}\left( m{\dot{r}}^{3}+6p|V|r^{5}e^{\frac{2t|V|}{r(t)}%
}\right) ^{2}}{64p^{2}|V|^{2}r^{8}}\right] .  \label{perturbation}
\end{equation}%
It follows that the unperturbed Lagrangian is 
\begin{equation}
L_{0}=g_{22}{\dot{\phi}}^{2}+g_{11}{\dot{r}}^{2}-H_{YM}=-\frac{\dot{\varphi}%
^{2}\left( \frac{m^{2}e^{-\frac{2t|V|}{r}}{\dot{r}}^{3}}{4p|V|r^{4}}+\frac{3%
}{2}mr\right) ^{2}}{4m}+\frac{1}{2}mr^{2}\dot{\varphi}^{2}+\frac{1}{2}{\dot{r%
}}^{2}\left( m-\frac{2p|V|r^{5}e^{\frac{2t|V|}{r}}}{{\dot{r}}^{3}}\right) .
\label{free-Lagrangian}
\end{equation}

Let us determine now an unperturbed semispray 
%Определим невозмущенный тензор
${G_{0}}_{(1)1}^{i}$ as a semispray satisfying the following system of
equations:% как тензор,
%удовлетворяющий системе уравнений
\begin{equation}
{\dfrac{d{y_{0}}_{1}^{i}}{dt}}+2{G_{0}}_{(1)1}^{(i)}\left( t,x_{0}^{k},{y_{0}%
}_{1}^{k}\right) =0,\qquad \dfrac{dx_{0}^{k}}{dt}={y_{0}}_{1}^{k},
\label{geodesic-equation1}
\end{equation}%
where 
\begin{equation}
\begin{split}
{G_{0}}_{(1)1}^{(i)}& =\left. \dfrac{g^{is}}{4}\left[ \dfrac{\partial
^{2}(L_{0}+\delta L)}{\partial x^{q}\partial y_{1}^{s}}y_{1}^{q}-\dfrac{%
\partial (L_{0}+\delta L)}{\partial x^{s}}+\frac{\partial ^{2}(L_{0}+\delta
L)}{\partial t\partial y_{1}^{s}}\right] \right\vert _{{\vec{r}}_{0},\vec{%
\dot{r}}_{0}}\equiv \\
& \equiv \dfrac{g^{is}}{4}\left[ \dfrac{\partial ^{2}L_{0}}{\partial
x^{q}\partial y_{1}^{s}}y_{1}^{q}-\dfrac{\partial L_{0}}{\partial x^{s}}+%
\frac{\partial ^{2}L_{0}}{\partial t\partial y_{1}^{s}}\right] +\mathcal{O}%
(d^{0}\delta L).
\end{split}
\label{add1to-geodesic-equation1}
\end{equation}%
A calculation of a perturbation %Расчет возмущение
$\delta G_{(1)1}^{i}$ of the semispray %тензора
$G_{(1)1}^{i}={G_{0}}_{(1)1}^{i}+\delta G_{(1)1}^{i}$ gives %дает
the expression 
\begin{equation*}
\delta G_{(1)1}^{i}=\dfrac{g^{is}}{4}\left[ \dfrac{\partial ^{2}L_{0}}{%
\partial x^{q}\partial y_{1}^{s}}\delta y_{1}^{q}+\dfrac{\partial ^{2}\delta
L}{\partial x^{q}\partial y_{1}^{s}}y_{1}^{q}-\dfrac{\partial \delta L}{%
\partial x^{s}}+\frac{\partial ^{2}\delta L}{\partial t\partial y_{1}^{s}}%
\right] -\mathcal{O}(d^{0}\delta L),
\end{equation*}%
which leads to equations %приводит к уравнению
\begin{equation}
\begin{array}{l}
{\dfrac{d\delta y_{1}^{i}}{dt}}+\dfrac{g^{is}}{2}\left. \dfrac{\partial
^{2}L_{0}}{\partial x^{q}\partial y_{1}^{s}}\right\vert _{\vec{r}_{0},\vec{%
\dot{r}}_{0}}\delta y_{1}^{q}+\medskip \\ 
+\dfrac{g^{is}}{2}\left\{ \left. \left[ \dfrac{\partial ^{2}}{\partial
x^{q}\partial y_{1}^{s}}\left( \dfrac{\partial \delta L}{\partial x^{\kappa }%
}\right) \right] y_{1}^{q}\right\vert _{\vec{r}_{0},\vec{\dot{r}}_{0}}\delta
x^{\kappa }+\left. \left[ \dfrac{\partial ^{2}}{\partial x^{q}\partial
y_{1}^{s}}\left( \dfrac{\partial \delta L}{\partial y_{1}^{\kappa }}\right) %
\right] y_{1}^{q}\right\vert _{\vec{r}_{0},\vec{\dot{r}}_{0}}\delta
y_{1}^{\kappa }-\left. \left[ \dfrac{\partial }{\partial x^{s}}\dfrac{%
\partial \delta L}{\partial x^{\kappa }}\right] \right\vert _{\vec{r}_{0},%
\vec{\dot{r}}_{0}}\delta x^{\kappa }-\right. \medskip \\ 
\left. -\left. \left[ \dfrac{\partial }{\partial x^{s}}\dfrac{\partial
\delta L}{\partial y_{1}^{\kappa }}\right] \right\vert _{\vec{r}_{0},\vec{%
\dot{r}}_{0}}\delta y_{1}^{\kappa }+\left. \left[ \dfrac{\partial ^{2}}{%
\partial t\partial y_{1}^{s}}\left( \dfrac{\partial \delta L}{\partial
x^{\kappa }}\right) \right] \right\vert _{\vec{r}_{0},\vec{\dot{r}}%
_{0}}\delta x^{\kappa }+\left. \left[ \dfrac{\partial ^{2}}{\partial
t\partial y_{1}^{s}}\left( \dfrac{\partial \delta L}{\partial y_{1}^{\kappa }%
}\right) \right] \right\vert _{\vec{r}_{0},\vec{\dot{r}}_{0}}\delta
y_{1}^{\kappa }\right\} =0,%
\end{array}
\label{geodesic-equation2-0}
\end{equation}%
\begin{equation}
\dfrac{d\delta x^{k}}{dt}=\delta y_{1}^{k}.  \label{geodesic-equation2}
\end{equation}%
Since %Так как
due to condition (\ref{singular-differenc-2D-Lagrangian}) $\dot{\varphi}%
_{0}=\epsilon $ is very small: $\dot{\varphi}_{0}=\epsilon \ll 1$ and %и
$\delta L$ depends quadratically on %зависит квадратично от
$\dot{\varphi}$, then the variables in %то переменные в
(\ref{geodesic-equation2-0}, \ref{geodesic-equation2}) are separated: 
%разделяются
%находим
\begin{equation*}
{\frac{d\delta y_{1}^{2}}{dt}}+\dfrac{g^{22}}{2}\left. \left\{ \left[ \dfrac{%
\partial ^{2}}{\partial x^{1}\partial y_{1}^{2}}\left( \frac{\partial \delta
L}{\partial y_{1}^{2}}\right) \right] y_{1}^{1}+\left[ \frac{\partial ^{2}}{%
\partial t\partial y_{1}^{2}}\left( \frac{\partial \delta L}{\partial
y_{1}^{2}}\right) \right] \right\} \right\vert _{r_{0},\dot{r}_{0}}\delta
y_{1}^{2}=0,
\end{equation*}%
\begin{equation*}
\begin{array}{l}
{\dfrac{d\delta y_{1}^{1}}{dt}}-\dfrac{g^{11}}{2}\left. \left[ \dfrac{%
\partial }{\partial x^{1}}\dfrac{\partial \delta L}{\partial x^{1}}\right]
\right\vert _{r_{0},{\dot{r}}_{0}}\delta x^{1}+\dfrac{g^{11}}{2}\left.
\left\{ \left[ \dfrac{\partial ^{2}}{\partial x^{q}\partial y_{1}^{1}}\left( 
\dfrac{\partial \delta L}{\partial y_{1}^{2}}\right) \right] y_{1}^{q}+\left[
\dfrac{\partial ^{2}}{\partial t\partial y_{1}^{1}}\left( \dfrac{\partial
\delta L}{\partial y_{1}^{2}}\right) \right] \right\} \right\vert
_{r_{0},\delta x^{2},{\dot{r}}_{0},\delta y_{1}^{2}}\delta y_{1}^{2}+\medskip
\\ 
+\dfrac{g^{11}}{2}\left. \dfrac{\partial ^{2}L_{0}}{\partial x^{1}\partial
y_{1}^{1}}\right\vert _{\vec{r}_{0},\vec{\dot{r}}_{0}}\delta
y_{1}^{1}=0,\qquad \dfrac{d\delta x^{k}}{dt}=\delta y_{1}^{k}.%
\end{array}%
\end{equation*}

Let us examine the case %Рассмотрим случай
$\dot{\varphi}_{0}=\epsilon =0$. In this case, a full separation 
%В этом случае имеем полное разделение
of variables takes place: %переменных:
\begin{eqnarray}
{\frac{d\delta y_{1}^{2}}{dt}}+\dfrac{g^{22}}{2}\left. \left\{ \left[ \dfrac{%
\partial ^{2}}{\partial x^{1}\partial y_{1}^{2}}\left( \frac{\partial \delta
L}{\partial y_{1}^{2}}\right) \right] y_{1}^{1}+\left[ \frac{\partial ^{2}}{%
\partial t\partial y_{1}^{2}}\left( \frac{\partial \delta L}{\partial
y_{1}^{2}}\right) \right] \right\} \right\vert _{r_{0},\dot{r}_{0}}\delta
y_{1}^{2} &=&0,  \label{geodesic-equation-final-a} \\
{\frac{d\delta y_{1}^{1}}{dt}}-\dfrac{g^{11}}{2}\left. \left[ \dfrac{%
\partial }{\partial x^{1}}\dfrac{\partial \delta L}{\partial x^{1}}\right]
\right\vert _{r_{0},{\dot{r}}_{0}}\delta x^{1}+\dfrac{g^{11}}{2}\left. 
\dfrac{\partial ^{2}L_{0}}{\partial x^{1}\partial y_{1}^{1}}\right\vert _{%
\vec{r}_{0},\vec{\dot{r}}_{0}}\delta y_{1}^{1} &=&0,
\label{geodesic-equation-final-b} \\
\frac{d\delta x^{k}}{dt} &=&\delta y_{1}^{k}.
\label{geodesic-equation-final-c}
\end{eqnarray}

Substituting %Подставляя
(\ref{perturbation}) and (\ref{free-Lagrangian}) into (\ref%
{geodesic-equation-final-a} -- \ref{geodesic-equation-final-c}) and making
some simple transformations, one gets 
%и делая простые преобразования, получаем уравнения для девиации в явном
the equations for deviation in explicit form: %виде:
%Приведение подобных членов
%\begin{eqnarray}
%256 p^4 V^4 r_0^{20}  e^{{8 t V}/{r_0}}
%\left(2 p V r_0^5 e^{{2 t V}/{r_0}} - m \dot r_0^3\right){d^2\over dt^2} \delta r(t)
%+256 p^5 V^5 r_0^{23}   e^{{10 t V}/{r_0}} (2 t V-5 r_0) \dot{r}_0
%{d\over dt}\delta r (t)  \nonumber \\
%+256 p^4 V^4 {\delta r} r_0^{20} e^{{8 t V}/{r_0}} \dot{r}_0^3
%{d^2\over dt^2}{U(t,r_0)} =0\end{eqnarray}
\begin{eqnarray}
\left( 2p|V|r_{0}^{5}e^{\frac{2t|V|}{r_{0}}}-m\dot{r}_{0}^{3}\right) {\frac{%
d^{2}}{dt^{2}}}\delta r(t)+p|V|r_{0}^{3}e^{\frac{2t|V|}{r_{0}}}(2t|V|-5r_{0})%
\dot{r}_{0}{\frac{d}{dt}}\delta r(t)+\dot{r}_{0}^{3}\ddot{U}(t,r_{0})\ {%
\delta r} &=&0,  \label{geodesic-equation-final1} \\
{\frac{d^{2}}{dt^{2}}}\delta \varphi (t)+\frac{\dot{r}_{0}\left( m\dot{r}%
_{0}^{3}+6p|V|r_{0}^{5}e^{\frac{2t|V|}{r_{0}}}\right) \left[ m(t|V|-2r_{0})%
\dot{r}_{0}^{3}+3p|V|r_{0}^{6}e^{\frac{2t|V|}{r_{0}}}\right] }{%
8p^{2}|V|^{2}r_{0}^{12}\exp \{{4t|V|}/{r_{0}}\}}{\frac{d}{dt}}\delta \varphi
(t) &=&0.  \label{geodesic-equation-final2}
\end{eqnarray}%
According to equation %Согласно уравнению
(\ref{resonance-trajectory}) the resonant solution %резонансное решение
$r_{0}(-t)$ satisfies the equation% удовлетворяет уравнению
\begin{equation}
m\dot{r}_{0}^{3}+6p|V|r_{0}^{5}e^{\frac{2t|V|}{r_{0}}}=0.
\label{resonance-trajectory1}
\end{equation}%
The substitution of %Подстановка
(\ref{resonance-trajectory1}) into (\ref{geodesic-equation-final2}) gives
the equation %, дает уравнение
\begin{equation*}
{\frac{d^{2}}{dt^{2}}}\delta \varphi (t)=0,
\end{equation*}%
that can be easy solved as % которое легко решить:
\begin{equation}
\delta \varphi =C_{1}+C_{2}t.  \label{solution-delta-phi}
\end{equation}

Further a simulation will be performed. 
%Далее проведем численное моделирование. Выберем следующие
Let us choose the following particular values for our parameters: 
%значения параметров:
$p=10,\ |V|=1000,\ m=1$ and assume that the constants in %, и положим
%константы в
(\ref{solution-delta-phi}) are equal to %, равными
$C_{1}=0,\ C_{2}=1$. Using the physical condition %Учитывая условие
$r_{0}\rightarrow R_{0}-|V|t\rightarrow 0$ for large time $t$, the resonant
solution %, резонансное решение
$r_{0}(-t)$ has been found approximately by a transformation of the equation 
% находилось приближенно путем преобразования уравнения
(\ref{resonance-trajectory}) for large time %для больших времен
$t$ to the form %к виде
\begin{equation}
m\dot{r}_{0}^{3}+6p|V|r_{0}^{5}e^{{2|V|t}/{(R_{0}-|V|t)}}=0,\qquad 
\mbox{for
large }\ t.  \label{resonance-trajectory2}
\end{equation}%
The solution of equation %Решение уравнения
(\ref{resonance-trajectory2}) has the form 
\begin{equation}
\begin{split}
r_{0}& ={27\sqrt{m}{R_{0}}ve^{\frac{t|V|}{t|V|-{R_{0}}}}}\left\{ -4\sqrt[3]{6%
}\sqrt[3]{p}{R_{0}}^{5/3}e^{-\frac{2{R_{0}}/3}{{R_{0}}-t|V|}}\left[ f\left( 
\frac{2}{3}\right) -f\left( \frac{2{R_{0}}/3}{{R_{0}}-t|V|}\right) \right]
\right. \\
& \left. +e^{-\frac{2t|V|/3}{{R_{0}}-t|V|}}\left( 9\sqrt[3]{m}|V|^{2/3}+6%
\sqrt[3]{6}\sqrt[3]{p}{R_{0}}^{5/3}\right) -6\sqrt[3]{6}\sqrt[3]{p}{R_{0}}%
^{2/3}({R_{0}}-t|V|)\right\} ^{-3/2}\qquad \mbox{for large }\ t.
\end{split}
\label{resonance-trajectory-solution}
\end{equation}

The simulation of the case with %Численное моделиривание случая с
$\epsilon =0$ results %дало  результаты, представленные
from Figs.~2--4.

Finally, we consider it is important to note that the obtained results are
correlated with a behavior of the surface tension 
%Полученные результаты соотнесем с физической величиной -- поверхностным
%натяжением
$P$ of the monolayer. Moreover, the dependence %монослоя. Зависимость
$r(t)$ shown in Fig. 4 behaves as an isotherm % , изображенная
%на последнем рисунке, ведет себя как изотерма
$P(S)$. It has a plateau corresponding to the first-order phase transition
and for large %  Она имеет плато,
%соответствующее фазовому переходу 1-го рода и загиб при больших
$t$ to a fold that corresponds to the collapse of monolayer.%  ,
%соответствующий коллапсу монослоя.

\textbf{Acknowledgements. }The present work was developed under the auspices
of Grant 1196/2012 - BRFFR-RA F12RA-002, within the cooperation framework
between Romanian Academy and Belarusian Republican Foundation for
Fundamental Research.

A version of this paper was presented at VIII-th International Conference
\textquotedblleft \textit{Finsler Extensions of Relativity Theory}%
\textquotedblleft , June 25 - July 1, 2012, Moscow-Fryazino, Russia.

%\newpage
%$P $ определяется
%через свободную энергию $F$ по формуле
%\begin{equation}
%P = \left.{\partial F\over \partial S}\right|_T
%\end{equation}
%where $S$ is a monolayer area, T is a temperature.

\end{document}